\documentclass[a4paper,fleqn]{article}

\usepackage{geometry}

\usepackage[british]{babel}
\usepackage{amsthm, amssymb, mathtools, amsmath}
\usepackage{color}
\usepackage{amsbsy}
\usepackage{enumerate}
\usepackage{mdwlist}
\usepackage{bigints}

\usepackage{enumitem}

\usepackage[numbers]{natbib}

\newcommand{\B}[1]{\boldsymbol{#1}}

\newtheorem{lemma}{Lemma}

\newtheorem{remark}{Remark}

\usepackage{graphicx}
\usepackage{subfigure}
\usepackage{float}

\usepackage[labelfont=bf,font=small,up]{caption}

\usepackage{url, verbatim}

\usepackage[final]{changes}
\setremarkmarkup{(#1)}

\usepackage{etoolbox}
\makeatletter
\patchcmd{\maketitle}{\@fnsymbol}{\@alph}{}{}  
\makeatother

\usepackage{authblk}

\title{Mean field at distance one}
\date{\today}
\author[,1,2,3]{Ka Yin Leung \thanks{kayin.leung@math.su.se}}
\author[,2,4]{Mirjam Kretzschmar\thanks{m.e.e.kretzschmar@umcutrecht.nl}}
\author[,1]{Odo Diekmann\thanks{o.diekmann@uu.nl}}

\affil[1]{\footnotesize Mathematical Institute, Utrecht University, Utrecht, The Netherlands}
\affil[2]{\footnotesize Julius Center for Primary Care and Health Sciences, University Medical Center Utrecht, Utrecht, The Netherlands}
\affil[3]{\footnotesize Department of Mathematics, Stockholm University, Stockholm, Sweden}
\affil[4]{\footnotesize National Institute for Public Health and the Environment, Bilthoven, The Netherlands}

\begin{document}
\maketitle
\begin{abstract}
To be able to understand how infectious diseases spread on networks, it is important to understand the network structure itself in the absence of infection. In this text we consider dynamic network models that are inspired by the (static) configuration network. The networks are described by population-level averages such as the fraction of the population with $k$ partners, $k=0,1,2,\ldots$ This means that the bookkeeping contains information about individuals and their partners, but no information about partners of partners. Can we average over the population to obtain information about partners of partners? The answer is `it depends', and this is where the mean field at distance one assumption comes into play. In this text we explain that, yes, we may average over the population (in the right way) in the static network. Moreover, we provide evidence in support of a positive answer for the network model that is dynamic due to partnership changes. If, however, we additionally allow for demographic changes, dependencies between partners arise. In earlier work we used the slogan `mean field at distance one' as a justification of simply ignoring the dependencies. Here we discuss the subtleties that come with the mean field at distance one assumption, especially when demography is involved. Particular attention is given to the accuracy of the approximation in the setting with demography. Next, the mean field at distance one assumption is discussed in the context of an infection superimposed on the network. We end with the conjecture that \deleted{a change in}\added{an extension of} the bookkeeping leads to an exact description of the network structure.
\end{abstract}

\section{Introduction}
Consider a large population of individuals who engage in partnerships. These partnerships make up the network structure of the population. The network evolves over time due to both demographic and partnership changes. Rather than keeping track of all individuals and partnerships over time, we are interested in a statistical description of the network at a particular point in time by characterizing population-level (p-level) quantities of interest, e.g.\ the fraction of the population having $k$ partners, $k=0,1,2,\ldots$, and use these p-level averages to describe the disease dynamics in the population. In general, it is not possible to use such a statistical description to predict the future spread of the disease\added{. Indeed,} \deleted{because} the precise network structure \deleted{matters}\added{influences how the disease is transmitted on the network. A statistical description of p-level fractions generally does not provide enough information to recover the structure of the network.} But\added{,} by making assumptions about the structure of the network, \added{e.g.\ by assuming a (static) configuration network, such a statistical description may be possible.} \deleted{such a description is possible, see our discussion in~\cite{Leung2016}.}

\added{The construction of the (static) configuration network guarantees the absence of degree-degree correlation. As a consequence it is easy to describe the transmission dynamics of an infectious disease across the network in the course of time, even for rather general infectivity functions, see~\cite[Section 2.5]{Leung2016} and~\cite{Barbour2013}. The dynamic network models that we consider in this text (and previous work~\cite{Leung2016}) are inspired by the (static) configuration network.}

\deleted{The dynamic network models that we consider in this text (and previous work) are inspired by the (static) configuration network. The construction of this network guarantees the absence of degree-degree correlation. As a consequence it is easy to describe the transmission dynamics of an infectious disease across the network in the course of time, even for rather general infectivity functions,} \deleted{see}\cite[Section 2.5]{Leung2016} \deleted{and} \cite{Barbour2013}.

An essential feature of the network models under consideration are the `binding sites' (in the static setting these are often referred to as half-edges or stubs; these were most cleverly used to describe transmission dynamics of an SIR infection in the static setting by Volz in~\cite{Volz2008} and in subsequent work by Miller, Volz, and coauthors (e.g.~\cite{Kiss2017} and references therein)): each individual consists of a number of (conditionally) independent binding sites. In the static configuration network, binding sites are paired in a uniform way. Two binding sites that are paired form a partnership between their owners. Note that, while the construction of the network may lead to self-loops and multiple partnerships between the same individuals, the proportions are such that we may ignore these in the infinite population limit, see e.g.~\cite{Durrett2006, vdHofstad2015} for precise statements and proofs. 

We distinguish three different levels in the network: (i) binding sites, (ii) individuals, and (iii) the population. Systematic model formulation relates the three levels to each other. At the binding-site and individual level (i-level) we have a Markov chain description with p-level influences captured by environmental variables. We work in the large population limit, so the description at the p-level is deterministic. The binding sites are the essential building blocks of the model and allow us to understand the dynamics at the p-level. 

\added{In~\cite{Leung2016} we formulated models for the spread of an infectious disease over two dynamic variants of the static configuration network. Here we reconsider these dynamic networks while focusing on the mean field at distance one assumption.} \deleted{In this text we consider two dynamic variants of the static configuration network.} First, in Section~\ref{sec:nodemography}, we allow for partnership formation and separation. Next, in Section~\ref{sec:demography}, \deleted{also demographic turnover is considered} \added{we also incorporate demographic turnover.} \deleted{and the effect that such dynamics has at the p-level.} \deleted{In particular, our focus in this text is on partners of partners. What information about partners of partners can we deduce from a network where our bookkeeping contains only information about individuals and their partners?} \added{The key question, central to the mean field at distance one assumption, is: what information about partners of partners can we recover from a bookkeeping scheme that only contains information about individuals and their partners?}

In the case of a static \added{configuration} network this question is readily answered. There is \added{\emph{independence}} \deleted{independence} in the degrees of partners. Therefore, the probability that a partner has $k$ partners is simply obtained from the size-biased degree distribution, i.e.\ if $P_k$ is the probability that a randomly chosen individual has $k$ partners, then $kP_k/\sum_l lP_l$ is the probability that a randomly chosen partner has $k$ partners (where $\sum_{\added{l}} l P_l$ serves as a normalization constant). We show in Section~\ref{sec:nodemography} that this property is also shared by the dynamic network without demography, but, as we show in Section~\ref{sec:demography}, things are more subtle in the dynamic network with demography. In fact, in Section~\ref{sec:demography}, we show that degree dependencies arise as a result of age dependencies, and we quantify the dependency between the degrees of partners by the correlation coefficient. In Section~\ref{sec:meanfield} we discuss the mean field at distance one assumption in the context of an infectious disease superimposed on the network, and we explain where additional complications arise. Finally, in Section~\ref{sec:conclusion}, we end with a discussion and some conclusions. We conjecture that changing the bookkeeping of partners to include age allows for an exact description of the dynamic network with demography.

\section{Dynamic network without demography}\label{sec:nodemography}
\subsection{Model formulation}
\deleted{As a first dynamic variant of the static configuration network} \added{To obtain a first dynamic configuration network we assume}:
\begin{itemize}
\item occupied binding sites become free at rate $\sigma$
\item free binding sites form partnerships at rate $\rho F$ where $F$ is the fraction of binding sites that are free (mass action/ supply and demand)
\end{itemize}
\added{cf.~\cite[Section 3]{Leung2016}.} \added{Let the partnership capacity $n$ of an individual be the maximum number of partners it \emph{may} have at any given time. The actual number of partners (a.k.a. the degree) of the individual changes over time according to the per-binding-site rules specified in the two bullets above. In principle, the partnership capacity $n$ is a random variable with a specified distribution (with finite first and second moments). Here, for the sake of exposition, we assume that the distribution is concentrated in one point (also denoted by $n$). In other words, all individuals have exactly the same partnership capacity $n$. As a consequence, the degree of an individual follows a binomial distribution, see eq.~\eqref{eq:degree} below.} \deleted{For the sake of exposition we assume that each individual has a partnership capacity $n$, meaning that it has $n$ binding sites. In principle, as with the static configuration network, $n$ is a random variable with a distribution that has, by assumption, a finite second moment. In the dynamic network, the partnership capacity $n$ of an individual denotes the maximum number of partners it \emph{may} have at the same time, but not the \emph{actual} number of partners it has at a given time. The number of partners it has, also called the degree of the individual, instead arises from the network dynamics (for constant $n$ this is simply a binomial distribution, see~\eqref{eq:degree} below).}

The network dynamics \deleted{imply} \added{entail} that binding sites of an individual behave independently of one another in partnership changes. One can describe the dynamics of the fraction $F$ of free binding sites with the following ordinary differential equation (ODE):
\begin{equation*}
\frac{dF}{dt}=-\rho F^2+\sigma(1-F).
\end{equation*}
As a consequence, we find that $F$ stabilizes at a value characterized by the identity
\begin{equation}\label{eq:Fstable1}
F=\frac{\sigma}{\rho F+\sigma}\added{.}
\end{equation}
\deleted{s}\added{S}olving for $F$ in terms of $\sigma$ and $\rho$ yields
\begin{equation*}
F=\frac{\sqrt{\sigma(4\rho+\sigma)}-\deleted{(}\sigma\deleted{)}}{2\rho}.
\end{equation*} 
Therefore, we assume that $F$ is constant and satisfies \added{eq.}~\eqref{eq:Fstable1}. This assumption for $F$ ensures that the network structure is stationary even though the network itself is changing over time due to partnership dynamics. As a matter of fact, we assume that we start in stationarity.

\subsection{Independence in the degrees of partners}\label{sec:dynamicnodegree}
\begin{table}[H]
\begin{tabular}{l|l}
Variable & Description\\
\hline 
$(P_k)_{k=0}^n$ & Degree distribution for a randomly chosen individual, $\sum_{k=0}^n P_k=1$ \\ 
$(q_k)_{k=1}^n$ & Degree distribution for a newly acquired partner, $\sum_{k=1}^n q_k=1$ \\ 
$\varphi(\xi)=F$ & Probability that a binding site is free at time $\xi$ \\
&\quad after its owner acquired a partner at a different binding site\\
$(\pi_{k,l}(\xi))_{k,l=1}^n$ & Joint degree distribution of two partners \\
&\quad at time $\xi$ after partnership formation, $\sum_{k,l=1}^n\pi_{k,l}=1$
\end{tabular}  
\caption{Overview of distributions that are used in Section~\ref{sec:nodemography}}
\label{tbl:distributions2}
 \end{table}

We adopt the convention that the joint degree of two partners refers to the total number of partners of each of the individuals (including their known partner). We calculate the joint degree distribution $\pi_{k,l}(\xi)$ of two partners at partnership duration $\xi$, given that they remain partners for this period of time; see Table~\ref{tbl:distributions2} for an overview. If our derivations seem overly detailed, please bear in mind that \added{these details}\deleted{this} serve\deleted{s} to prepare for the analysis in Section~\ref{sec:demography} of a more subtle situation.

First, by combinatorics, we find that the probability  $P_k$ that a randomly chosen individual in the population has $k$ partners is simply 
\begin{equation}\label{eq:degree}
P_k=\binom n k F^{n-k}(1-F)^{k},
\end{equation}
i.e.\ the degree distribution $(P_k)$ is a binomial distribution with parameters $n$ and $1-F$. 

The probability $q_k$ that a newly acquired partner has $k$ partners (in total) is $(n-k+1)P_{k-1}/\sum_l (n-l)P_l$ (a potential partner in state $k-1$ has $(n-k+1)$ free binding sites; immediately after partner formation, it will be in state $k$). Here the sum serves to renormalize into a probability distribution. Working this out, we find that
\begin{align}
q_k&=\left.(n-k+1)\binom{n}{k-1}F^{n-k+1}(1-F)^{k-1}\middle/\sum_l(n-l)P_l\right.\nonumber\\
&=\left.n\binom{n-1}{k-1}F^{n-k+1}(1-F)^{k-1}\middle/nF\right.\nonumber\\
&=\binom{n-1}{k-1}F^{n-k}(1-F)^{k-1}\label{eq:q1}.
\end{align}
So, we find that a newly acquired partner has at least one occupied binding site and the other $n-1$ binding sites are free with probability $F$ and occupied with probability $1-F$. Hence, $q_k$ is equal to the probability that a randomly chosen \emph{partner} has $k$ partners. 

Next, let $\varphi$ denote the probability that a binding site is free at time $\xi$ after partner acquisition at another binding site of the same owner. Then $\varphi(\xi)=F$, since binding sites behave independently of one another. On the other hand, $\varphi$ satisfies
\begin{align}
\frac{d\varphi}{d\xi}&=-\rho F\varphi+\sigma(1-\varphi)\added{,}\label{eq:varphi1}\\
\varphi(0)&=F\nonumber\added{.}
\end{align}
Solving for $\varphi$ we find that 
\begin{equation}\label{eq:varphiF}
\varphi(\xi)=\frac{\sigma}{\rho F+\sigma}+\frac{\rho F^2-\sigma(1-F)}{\rho F+\sigma}e^{-(\rho F+\sigma)\xi}=F,
\end{equation}
where we used identity~\eqref{eq:Fstable1} for $F$ in the second equality. In particular, this confirms our intuition that partnership duration $\xi$ is not relevant.

We are now ready to consider the probability $\pi_{k,l}(\xi)$ that $u$ and $v$ have $k$ and $l$ partners in total at time $\xi$ after they formed a partnership, given that $u$ and $v$ remain partners in the period under consideration. First of all, note that 
\begin{equation}\label{eq:initial}
\pi_{k,l}(0)=q_kq_l,
\end{equation} 
with $q_j$ given by \added{eq.}~\eqref{eq:q1} (partnerships are formed at random between free binding sites in the population), i.e.\ there is independence in the degrees of the individuals in a newly formed partnership. Furthermore, both $u$ and $v$ have exactly one binding site occupied by their known partner and $n-1$ other binding sites at which partnership formation and separation can take place. Taking into account partnership-formation and -separation at the other binding sites, and conditioning on the existence of partnership $uv$ in the period under consideration, we find that $\pi_{k,l}$ satisfies
\begin{equation}\label{eq:ODEpi}
\begin{aligned}
\frac{d\pi_{k,l}}{d\xi}&=-\rho F(n-k)\pi_{k,l}-\rho F (n-l)\pi_{k,l}-\sigma(k-1)\pi_{k,l}-\sigma(l-1)\pi_{k,l}\\
&\phantom{=\ }+\rho F(n-k+1)\pi_{k-1,l}+\rho F(n-l+1)\pi_{k,l-1}+\sigma k \pi_{k+1,l}+\sigma l \pi_{k,l+1},
\end{aligned}
\end{equation}
with initial condition~\eqref{eq:initial}. Let
\begin{equation}\label{eq:pidef}
p_{k,l}(\xi)\coloneqq\binom{n-1}{k-1}\varphi(\xi)^{n-k}(1-\varphi(\xi))^{k-1}\binom{n-1}{l-1}\varphi(\xi)^{n-l}(1-\varphi(\xi))^{l-1}\added{.}
\end{equation}
We claim that $\pi_{k,l}(\xi)=p_{k,l}(\xi)$. Indeed, differentiating $p_{k,l}(\xi)$ with respect to $\xi$ and using \added{eq.}~\eqref{eq:varphi1}, we find that the ODE~\eqref{eq:ODEpi} for $\pi_{k,l}$ is indeed satisfied. Since these are straightforward calculations, we omit the details, and only note that the relations $(n-k)\binom{n-1}{k-1}=k\binom{n-1}{k}$ and $(k-1)\binom{n-1}{k-1}=(n-k+1)\binom{n-1}{k-2}$ yield the desired result. 

On the other hand, since $\varphi(\xi)=F$ (see \added{eq.}~\eqref{eq:varphiF}), from \added{eq.}~\eqref{eq:pidef}, we find that $p_{k,l}(\xi)=q_kq_l$. Hence the joint degree distribution of two partners at time $\xi$ after partnership formation is given by 
\begin{equation*}
\pi_{k,l}(\xi)=q_kq_l.
\end{equation*} 
In particular, we find that there is independence in the degrees of two partners. Moreover, the joint degree of two individuals in a partnership is the same (i) at partnership formation, (ii) at a specific partnership duration $\xi$ of their partnership, and (iii) at a randomly chosen time in their partnership. The only information that such a partnership gives us about the degree of the partners is that both partners have at least one occupied binding site.

\section{Dynamic network with demography}\label{sec:demography}
\subsection{Model formulation}\label{sec:dynamic2model}
A next dynamic variant of the static configuration network model is obtained by adding demographic turnover to the partnership changes of Section~\ref{sec:nodemography}, cf.~\cite{Leung2012} \added{and~\cite[Section 4]{Leung2016}}. We additionally assume that
\begin{itemize}
\item life length is exponentially distributed with parameter $\mu$
\item newborn individuals appear at a constant rate (which is equal to $\mu$ if we consider fractions, i.e.\ normalize the total population size to 1)
\item at birth, individuals enter the population without any partners
\end{itemize}
In other words, we assume a stationary age distribution with density $a\mapsto \mu e^{-\mu a}$. Note that our assumptions on demography imply that the rate at which occupied binding sites become free is $\sigma+\mu$ where $\sigma$ corresponds to `separation' and $\mu$ to `death of the partner'. The fraction $F$ of binding sites that are free now satisfies
\begin{equation*}
\frac{dF}{dt}=\mu-\mu F-\rho F^2+(\sigma+\mu)(1-F).
\end{equation*}
Again, as in Section~\ref{sec:nodemography}, $F$ stabilizes to a constant that satisfies
\begin{equation}\label{eq:Fconstant}
\rho F^2-(\sigma+2\mu)(1-F)=0.
\end{equation} 
Therefore, we assume that $F$ is constant. In terms of the model parameters this constant $F$ equals 
\begin{equation}\label{eq:Ffunction}
F=\frac{\sqrt{(\sigma+2\mu)(4\rho+\sigma+2\mu)}-(\sigma+2\mu)}{2\rho}.
\end{equation} 
As a consequence, although the network itself changes due to partnership dynamics and demographic turnover, \added{the population structure is statistically stable. In particular, the degree distribution does not change with time.} \deleted{the network structure is stable over time (and we assume that we start in stationarity).} \added{We assume that the network starts in stationarity.}

In subsequent subsections we show that, contrary to the static network and the dynamic network without demography, there is no longer independence in degrees in the dynamic network with demography. We do so by showing that age-age dependence and age-degree dependence exist, leading to the conjecture that demographic turnover causes degree-degree dependence. We verify the conjecture, and \deleted{quantify the dependence with a} \added{compute the} correlation coefficient and study how it depends on the parameters. For convenience an overview of the probabilities and distributions relevant \added{to} \deleted{for} Section~\ref{sec:demography} are given in Table~\ref{tbl:distributions3}. A related dynamic network incorporating demography (in a growing population) is considered in~\cite{Britton2010,Britton2011}, where also the degree-degree correlation is \deleted{assessed} \added{determined}.

Finally, as always, the assumptions matter. Here we model demography so that individuals enter the population without any partners. After its birth, an individual may acquire and lose partners according to the rules \deleted{of} \added{assumed in}~Section~\ref{sec:nodemography}. Therefore, the number of partners of an individual contains information about the age of that individual. One may think of different ways of modelling demography that may not necessarily lead to age dependencies. Indeed, an assumption for the numbers of partners of newborn individuals made in~\cite{Lashari2016} achieves that age dependencies are absent.

\begin{table}[H]
\begin{tabular}{l|l}
Variable & Description\\
\hline 
$\varphi(a)$ & Probability that a binding site is free given that its owner has age $a$\\ 
$\pi_0(a)$ & Density function for the age of the owner of a randomly chosen free binding site \\ 
$\pi_1(a)$ & Density function for the age of the owner\\
&\quad  of a randomly chosen occupied binding site\\
$H(a,\alpha)$ & Density function for the ages of two partners in a randomly chosen partnership\\
$p_k(a)$ & Probability that an individual has $k$ partners at age $a$, $k=0,\ldots,n$\\
$q_k(a)$ & Probability that a partner of age $a$ has $k$ partners in total, $k=1,\ldots,n$\\
$P(k,l)$ & Probability that the joint degree of a randomly chosen partnership is $(k,l)$, \\
&\quad$k,l=1,\ldots,n$
\end{tabular}  
\caption{Overview of probabilities and densities that are used in Section~\ref{sec:demography}. By assumption, whenever we write `at age $a$' (or just `age $a$'), the individual under consideration remains alive in the period between being born and reaching age $a$.}
\label{tbl:distributions3}
 \end{table}

\subsection{Age-age and age-degree dependencies}\label{sec:ageage}
We show that there is dependence between the ages of partners by reasoning at the binding site and partnership level (compare with the derivation of the correlation coefficient for the related model in~\cite[Section 3.3]{Britton2011}). Whenever we write `at age $a$' (or just `age $a$'), the individual under consideration has, by assumption, survived until that age. First, consider a binding site (see also~\cite[Section 4]{Leung2016}). Let $\varphi(a)$ denote the probability that a binding site is free at age $a$. Then $\varphi$ satisfies the ODE 
\begin{equation*}
\frac{d\varphi}{da}=-\rho F\varphi+(\sigma+\mu)(1-\varphi),
\end{equation*}
with birth condition $\varphi(0)=1$, so
\begin{equation}\label{varphi}
\varphi(a)=\frac{\sigma+\mu}{\rho F+\sigma+\mu}+\frac{\rho F}{\rho F+\sigma+\mu}e^{-(\rho F+\sigma+\mu)a}.
\end{equation}
We have the identity
\begin{equation*}
F=\int_0^\infty \mu e^{-\mu a}\varphi(a)da
\end{equation*}
for the fraction of free binding sites (\added{one can }use \added{eq.}~\eqref{eq:Fconstant} \added{to check that this identity holds}). Then the probability density function of the age of (the owner of) a free binding site is given by
\begin{equation}\label{eq:agedistribution}
\pi_0(a)=\frac{\mu e^{-\mu a}\varphi(a)}F.
\end{equation}
Similarly, the probability density function for the age of a randomly chosen occupied binding site is 
\begin{equation}\label{eq:agedistribution2}
\pi_1(a)=\frac{\mu e^{-\mu a}(1-\varphi(a))}{1-F}.
\end{equation}
Next, observe that, due to independence at partner formation, the joint age density function of two partners at partner formation is the product $\pi_0(a)\pi_0(\alpha)$ of age density functions of free binding sites. Two free binding sites are paired at rate $\rho F^2$, and a partnership dissolves at rate $\sigma+2\mu$. Furthermore, newborn individuals (at age 0) have no partners. Therefore, the density function for the ages of two partners in a randomly chosen partnership satisfies
\begin{align*}
\frac{\partial p}{\partial a}+\frac{\partial p}{\partial \alpha}&=\rho F^2 \pi_0(a)\pi_0(\alpha)-(\sigma+2\mu)p\added{,}\\
p(0,\alpha)&=0=p(a,0)\added{.}
\end{align*}
Solving for $p$ and normalizing into a probability density function $H(a,\alpha)$ for the ages of two partners in a randomly chosen partnership yields
\begin{align}
H(a,\alpha)&=\frac{p(a,\alpha)}{\int_0^\infty\int_0^\infty p(b,\beta)dbd\beta}\nonumber\\
&=\frac{\int_0^{\min{(a,\alpha)}}\rho F^2\pi_0(a-\xi)\pi_0(\alpha-\xi)e^{-(\sigma+2\mu)\xi}d\xi}{1-F}.\label{eq:H}
\end{align}
Here we used that $\int_{a=0}^\infty\int_{\alpha=0}^\infty p(a,\alpha)dad\alpha=1-F$ in the second equality (use \added{eq.}~\eqref{eq:probint} in Lemma~\ref{lem:probint} of Appendix~\ref{app:probint}).

Note that we can also reason directly from the interpretation \added{of the model} to obtain \added{eq.}~\eqref{eq:H}. Consider a randomly chosen partnership of duration $\xi$ with partners of age $a$ and $\alpha$, then at partnership formation these individuals had ages $a-\xi$ and $\alpha-\xi$. At partnership formation, their ages are independent and the densities are $\pi_0(a-\xi)$ and $\pi_0(\alpha-\xi)$. The rate at which a partnership is formed is $\rho F^2$. Next, the probability that a partnership has duration of at least $\xi$ is $e^{-(\sigma+2\mu)\xi}$. Obviously, the partnership duration $\xi$ satisfies $0\leq\xi\leq\min(a,\alpha)$. Finally, the probability that a binding site is occupied is $1-F$, yielding the normalizing constant. By combining these elements we obtain \added{eq.}~\eqref{eq:H}.

Finally, the expression for $H(a,\alpha)$ allows us to conclude that there is dependence in the ages of two partners. Indeed, if these were independent of one another, then the probability density function for the ages of two partners in a randomly chosen partnership would be the product of the probability density functions for the age of a randomly chosen occupied binding site, i.e.\ $\pi_1(a)\pi_1(\alpha)$ with $\pi_1$ given by \added{eq.}~\eqref{eq:agedistribution2}. Since $H(a,\alpha)\neq\pi_1(a)\pi_1(\alpha)$, we conclude that demographic turnover (in the way that 
we have modelled it) induces age dependence. 

In order to show that the age and degree of an individual are correlated, we relate the binding site level to the i-level. Let $p_k(a)$ denote the probability that an individual has $k$ partners at age $a$. Then, by the independence assumption for binding sites, combinatorics yields
\begin{equation*}
p_k(a)=\binom{n}{k}\varphi(a)^{n-k}(1-\varphi(a))^{k}.
\end{equation*}
In particular, we find that information about the age of an individual helps to predict its degree. Since we also have dependence in the ages of two partners, we expect that there is dependence in the degrees of two partners. This dependence is quantified by means of the degree correlation coefficient in the next subsection.

\subsection{Quantifying the \added{degree-}degree dependence}\label{sec:degreedegree}
\subsubsection{Joint degree distribution}
Let $P(k,l)$ denote the probability that a randomly chosen partnership has joint degree $(k,l)$, $1\leq k,l\leq n$, i.e.\ the probability that two individuals have degrees $k$ and $l$ given that they are partners.

Next, we consider the probability that an individual $u$ has $k$ partners at age $a$, given that it is a partner of an individual $v$ with age $\alpha$. By assumption, given the age of the owner, the binding sites of an individual are independent of one another as long as the owner does not die. Therefore, the fact that $u$ of age $a$ is in a partnership with individual $v$ simply means that one of the $n$ binding sites of $u$ is already occupied. The probability that any other binding site of $u$ is free is $\varphi(a)$ with $\varphi(a)$ given by \added{eq.}~\eqref{varphi}. Combinatorics yields that the probability that $u$ has, at age $a$, $k$ partners in total, given partner $v$ with age $\alpha$, is equal to
\begin{equation}\label{eq:qk_age}
q_k(a)=\binom{n-1}{k-1}\varphi(a)^{n-k}(1-\varphi(a))^{k-1},
\end{equation}
\added{where} $1\leq k\leq n$, $a>0$. Conditioning on $u$ having age $a$ and $v$ having age $\alpha$, the probability that the joint degree of $u$ and $v$ is $(k,l)$ is simply
\begin{equation}\label{eq:qk}
q_k(a)q_l(\alpha).
\end{equation}
The probability density function for the ages of the partners in a randomly chosen partnership is $H(a,\alpha)$ (see \added{eq.}~\eqref{eq:H}). By integrating over all possible ages $a$ and $\alpha$ we obtain the probability $P(k,l)$ that two individuals $u$ and $v$ in a randomly chosen partnership have $k$ and $l$ partners:
\begin{align}
P(k,l)&=\int_{\alpha=0}^\infty\int_{a=0}^\infty q_k(a)q_l(\alpha)H(a,\alpha)dad\alpha\label{exact}
\end{align}
Note that two individuals in a randomly chosen partnership are identically distributed (with respect to age as well as with respect to number of partners) so $P$ is symmetric in $k$ and $l$, i.e.\ $P(k,l)=P(l,k)$. Furthermore, note that both the $q_k(a)$ and $\pi_0(a)$ are functions of $\varphi(a)$ with $\varphi(a)$ given by \added{eq.}~\eqref{varphi}. By algebraic expansion of the powers of the form $(x+y)^m$ one can rewrite these probabilities $q_k$ as
\begin{align*}
q_k(a)&=\binom{n-1}{k-1}\sum_{j=0}^{n-k}\sum_{i=0}^{k-1}\binom{n-k}{j}\binom{k-1}{i}(-1)^i \\
&\phantom{\binom{n-1}{k-1}\sum_{j=0}^{n-k}}\left(\frac{\sigma+\mu}{\rho F+\sigma+\mu}\right)^{k-1+j}\left(\frac{\rho F}{\rho F+\sigma+\mu}\right)^{n-k-j}e^{-(\rho F+\sigma+\mu)(i+j)a}.
\end{align*} 
So \added{eq.}~\eqref{exact} can be written as the sum of integrals over exponential functions. By working out these integrals one obtains an explicit expression for $P(k,l)$ in terms of the model parameters. \deleted{In Section~\ref{sec:dependence}, when investigating the dependence between the number of partners of two partners, the sum $\sum kl P(k,l)$ plays an important role.} \added{In this paper, our aim is to quantify the dependence between the degrees of partners by means of the correlation coefficient, and} \deleted{As} we are not that much interested in \added{the specific} probabilities $P(k,l)$ for specific $k$ and $l$. \added{Therefore}, we do not write down the explicit expression for $P(k,l)$. \added{Rather, when considering the correlation coefficient in Section~\ref{sec:dependence}, the sum $\sum kl P(k,l)$ plays an important role, which, among other sums, is worked out in Appendix~\ref{app:corr}.} Finally, we can obtain the marginal degree distribution $(Q_k)$ from the joint degree distribution $P(k,l)$, and this leads to the same expression for $Q_k$ as the one derived from the degree distribution in the population (see our previous work~\cite{Leung2012}); see Appendix~\ref{app:marginal} for details.

\begin{remark}[Binding site or i-level perspective]
Note that one could also obtain the probability $P(k,l)$ by taking the perspective of individuals in a partnership rather than binding sites. This is exactly what we have done in~\cite[Appendix B]{Leung2015a} for $n=2$. We calculated probability~\eqref{exact} for $n=2$ using the number of partners of two individuals in a partnership without taking into account their ages. While in principle this is not more difficult than the reasoning in this section, generalizing to $n>2$ quickly becomes quite involved as the number of possible $(k,l)$ pairs grows quickly with $n$. This is also the reason that we only worked out $n=2$ in the appendix of~\cite{Leung2015a}: for $n=2$ only the inverse of a $3\times3$ matrix is needed, but the size of this matrix quickly grows with $n$. Nevertheless, one can check that \added{eq.}~\eqref{exact} coincides with \added{eq.}~(51) of~\cite{Leung2015a} for $n=2$ (e.g.\ by using Mathematica and identity (20) in~\cite{Leung2012} for $F$).
\end{remark}

\subsubsection{Correlation coefficient}\label{sec:dependence}
Choose a partnership at random and consider one of the \added{partners} \deleted{participants}. The probability that this individual has $k$ partners is $Q_k$ where $Q_k$ is given by \added{eq.}~\eqref{eq:Qk}. On the other hand, the probability that the randomly chosen partnership has joint degree $(k,l)$ is $P(k,l)$ with $P(k,l)$ given by \added{eq.}~\eqref{exact}. In~\cite{Leung2016,Leung2015a} we approximated the network structure by pretending that there is independence between partners of two individuals that are in a partnership, i.e.\ we approximated $P(k,l)$ by $Q_kQ_l$. That this is really an approximation, i.e.\ that $P(k,l)\neq Q_kQ_l$, was shown by way of explicit calculations for $n=2$ in~\cite[Appendix B]{Leung2015a}. In this section we investigate the approximation in more detail.  

We use the correlation coefficient to quantify the dependence (this coefficient is often denoted by $\rho$ but since we have already reserved this symbol for the \deleted{per-}partnership formation rate we will simply write $corr$). Let $D_u$ and $D_v$ be the random variables denoting the degrees of the individuals $u$ and $v$ in a randomly chosen partnership. Then the joint probability distribution of $D_u$ and $D_v$ is $(P(k,l))$. The degree correlation coefficient is given by
\begin{equation}
corr=\frac{E(D_uD_v)-E(D_u)E(D_v)}{\sqrt{Var(D_u)Var(D_v)}}=\frac{Cov(D_u,D_v)}{Var(D_u)}=\frac{A-B^2}{C-B^2}.\label{eq:corr}
\end{equation}
where $Cov(D_u,D_v)$ is the covariance of $D_u$ and $D_v$, and 
\begin{equation}\label{eq:sums}
\begin{aligned}
A=\sum_{k,l}klP(k,l),\qquad B=\sum_k kQ_k, \qquad\text{and}\qquad C=\sum_k k^2Q_k.
\end{aligned}
\end{equation}
The correlation coefficient satisfies \mbox{$-1\leq corr\leq1$}, where $corr=-1$ corresponds to fully disassortative mixing in the degrees of partners and $ corr=1$ corresponds to fully assortative mixing. In case the degrees of partners are independent of one another, the correlation coefficient is zero. Note that for $n=1$, i.e.\ in the case of monogamous pair formation, the degree of a partner is always 1, i.e.\ $Q_1=1$ and $P(1,1)=1$. The correlation coefficient is not defined for this case. We are only interested in $n=2,3,\ldots$

We are interested in the behaviour of $ corr$ as a function of the four model parameters $n$ (partnership capacity), $\sigma$ (partnership separation rate), $\rho$ (\deleted{per-pair} partnership formation rate), and $\mu$ (death rate). An explicit expression for $corr$ in terms of the model parameters is found by working out $A$, $B$, and $C$ (defined by \added{eq.}~\eqref{eq:sums}) at the right-hand side of \added{eq.}~\eqref{eq:corr}. These can all be expressed as integrals of simple functions of $\varphi(a)$, which we then can evaluate (note that $F$ is also a function of model parameters $\sigma$, $\rho$, and $\mu$; see \added{eq.}~\eqref{eq:Ffunction}). We work this out in Appendix~\ref{app:corr}. 

The calculations in~\cite[Appendix B]{Leung2015a} already showed that there is dependence for $n=2$ with $ corr>0$. So for $n=2$ there is assortativity with respect to the degrees of partners. For general $n$, the expression~\eqref{eq:covariance} in Appendix~\ref{app:corr} in this text shows that $corr>0$. So, in accordance with our expectation in Section~\ref{sec:ageage}, there is dependence in the degrees of partners. Moreover, $ corr>0$ for all $\rho,\sigma,\mu>0$, and $n>1$. In other words, for all $n>1$, the network is assortative in the degree: partners tend to have similar degrees. We study how $corr$ depends on the model parameters in the next subsection.

\begin{remark}[Limiting behaviour $\mu\to0$]
Note that $\lim_{\mu\to 0}corr=0$ (use the explicit expression for $corr$ calculated in Appendix~\ref{app:corr}), in complete accordance with the independence in degrees in the dynamic network model without demography of Section~\ref{sec:nodemography}.
\end{remark}

\begin{remark}[Degree-degree correlation]
\added{Degree correlation does occur in some real world networks and constructive procedures to generate networks with prescribed degree-degree correlation have been devised, see~\cite{Newman2002, Newman2003, Ball2012} and references therein. As has been shown here, a dynamic network model incorporating demographic turnover can exhibit degree-degree correlation as a consequence of age-age correlation. This provides a possible mechanistic interpretation of emergent assortative mixing.}
\end{remark}

\subsubsection{The effect of demographic changes on the correlation coefficient}\label{sec:numerics}
Now that we have an explicit expression for $ corr$, we can ask how it depends on the model parameters. Our main interest is in the relative time scales of demographic changes compared to partnership changes. Therefore, we are interested in \added{$\mu/\sigma\in(0,1]$.} \added{In particular, we are interested in} the limit $\mu/\sigma\to0$, i.e.\ in the limit that partnership changes are much faster than any demographic changes, while at the same time $\rho/\sigma$ remains constant, i.e.\ the \deleted{per-pair} \added{partnership} formation rate $\rho$ and the separation rate $\sigma$ are on the same time scale.

The formula~\eqref{eq:corr} for $corr$ and the expressions~\eqref{corr1}-\eqref{corr3} allow for an explicit expression of $corr$ in terms of model parameters $n$, $\sigma$, $\rho$, and $\mu$. To eliminate one parameter, we rewrite the correlation coefficient $corr$ as a function of $n$, $\tilde\rho=\rho/\sigma$, and $\tilde\mu=\mu/\sigma$:
\begin{equation*}
corr=\frac{a}{b},
\end{equation*}
with
\begin{align*}
a&=\tilde\mu^2 (n-1) (2 \tilde\mu +3 x+3) \left(4 \tilde\mu^2+2 \tilde\rho^2+\tilde\mu  (8 \tilde\rho -2 x+4)-2 \tilde\rho  (x-2)-x+1\right),\\
b&=2 (-2 \tilde\mu +x-1) (2 \tilde\mu +x+2)\big\{-4 \tilde\mu^3 (n+1)+2 \tilde\mu^2 (n (-3 \tilde\rho +x-2)-9 \tilde\rho +x-3\big\} \\
&\phantom{=\ }+\tilde\mu  (n (\tilde\rho  (x-3)+x-1)-23 \tilde\rho +7 \tilde\rho  x+2 x-2)+6 \tilde\rho  (x-1)),\\
x&=\sqrt{(1+2\tilde\mu)(4\tilde\rho+1+2\tilde\mu)}.
\end{align*}

Next, by considering the derivative of $ corr$ with respect to the parameters $n$, $\tilde\rho$, and $\tilde\mu$, we find that $ corr$ is strictly increasing in $n$ and $\tilde\rho$, and $\tilde\mu$. In Section~\ref{sec:dependence} we also observed that $corr>0$ so the network is assortative in degree. Furthermore, \added{since $corr$ is strictly increasing in $n$, $\tilde\rho$, and $\tilde\mu$, by considering the limit $\lim_{n\to\infty,\tilde\rho\to\infty, \tilde\mu\to1} corr=1/4$, we find that the correlation coefficient is at most 1/4, for all $n\geq2$, $\tilde\rho>0$, and $\tilde\mu\in(0,1]$.}\deleted{we find that the correlation coefficient has a supremum with value $1/4$, i.e.}
\begin{equation*}
\deleted{\sup\{ corr\mid \tilde\rho\geq0, 0\leq\tilde\mu\leq1, n\geq2\}=\lim_{n\to\infty,\tilde\rho\to\infty, \tilde\mu\to1} corr=\frac14.}
\end{equation*}
\deleted{So the correlation coefficient is always between 0 and 1/4. }

Finally, we investigate $ corr$ numerically. For fixed $n$, we investigate the correlation coefficient as a function of $\tilde\mu$ and $\tilde\rho$. \added{As we are interested in the relative time scales of demographic changes compared to partnership changes, we consider $\tilde\mu\in(0,1]$ and $\tilde\rho\in(0,\infty)$).} In general, we find that the correlation coefficient is close to zero; see Fig.~\ref{fig:assortative} for $n=3$ and $n=30$. So, although there is dependence, the dependence is generally not very strong. The correlation coefficient $ corr$ is largest when the time scales of demographic and partnership changes are close to each other. Moreover, the higher the \deleted{partner acquisition} \added{partnership formation} rate \added{is} compared to the separation rate, the larger $ corr$ is.

\begin{figure}[H]
\centering
\includegraphics[scale=0.6]{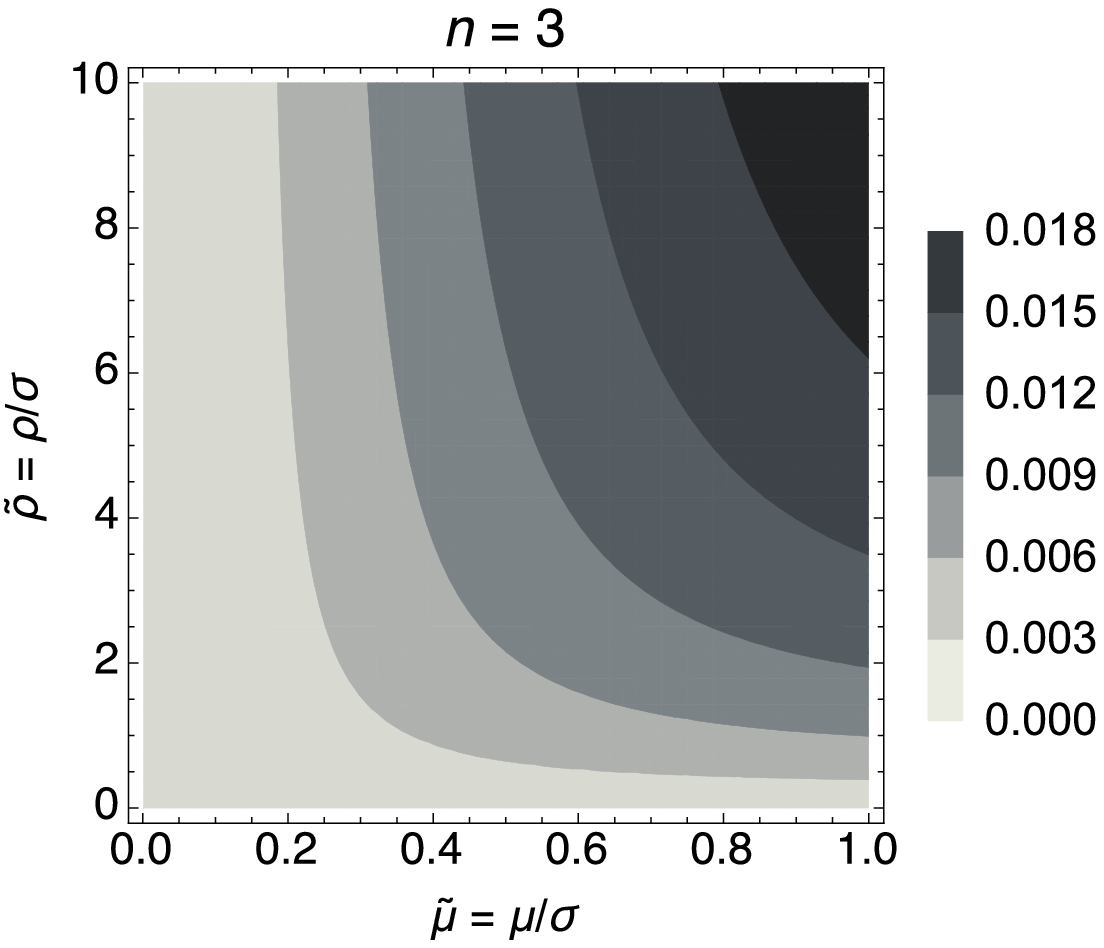}
\hfill
\includegraphics[scale=0.6]{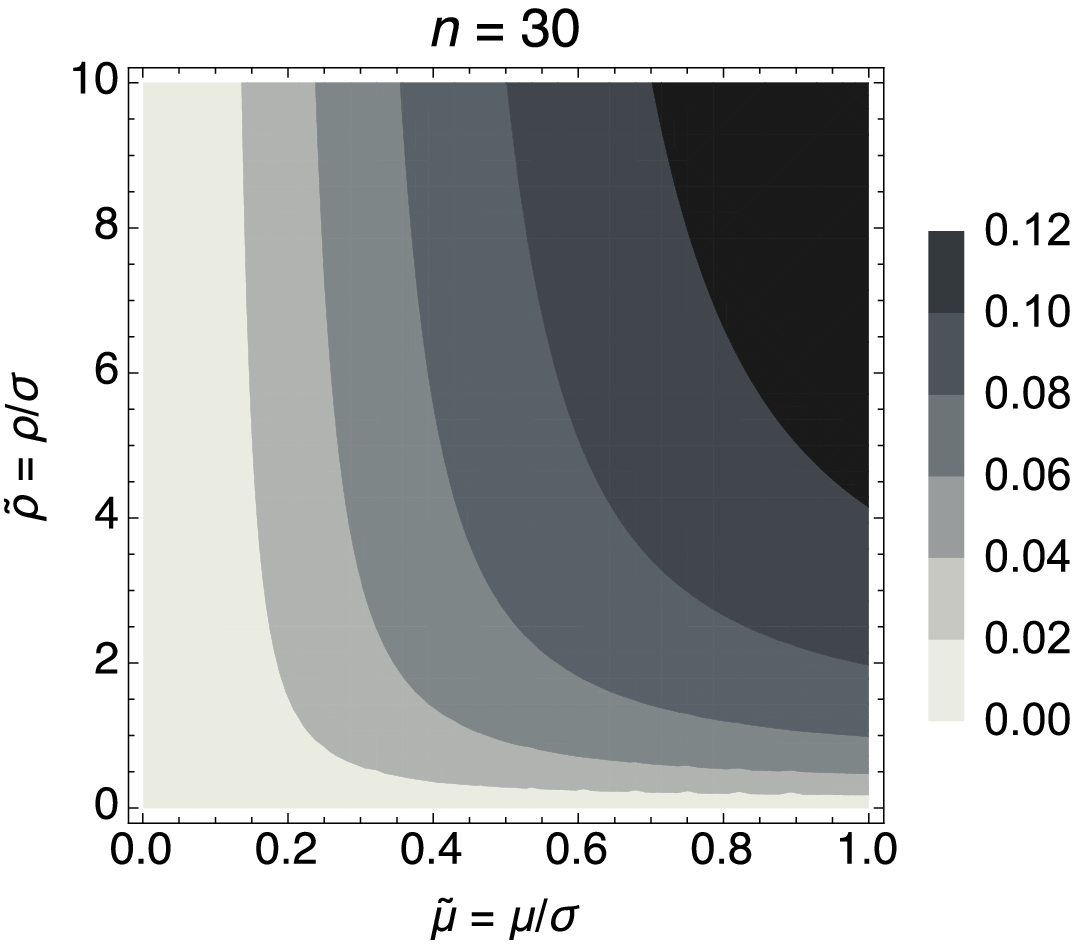}
\caption{Correlation coefficient as a function of $\tilde\mu$ (x-axis) and $\tilde\rho$ (y-axis) for $n=3$ and $n=30$. Note that the colour scales are different in the two figures: for $n=3$ the scale is between 0 and 0.018 and for $n=30$ the scale is between 0 and 0.12.}
\label{fig:assortative}
\end{figure}

Next, we compare \added{the effect of} different $n$ values \deleted{to each other} while keeping $\tilde\rho$ fixed in Fig.~\ref{fig:correlation}. While the correlation coefficient increases as a function of $n$, for relatively small values of $n$, $ corr$ remains relatively close to 0 (compared to the supremum value of 1/4). For all $n$, it holds that the faster partnership changes are compared to demographic changes, i.e.\ the smaller $\tilde\mu$ is, the smaller $ corr$ is. We find that $ corr$ tends to zero quite rapidly as $\tilde\mu\to0$.

\begin{figure}[H]
\centering
\includegraphics[scale=0.8]{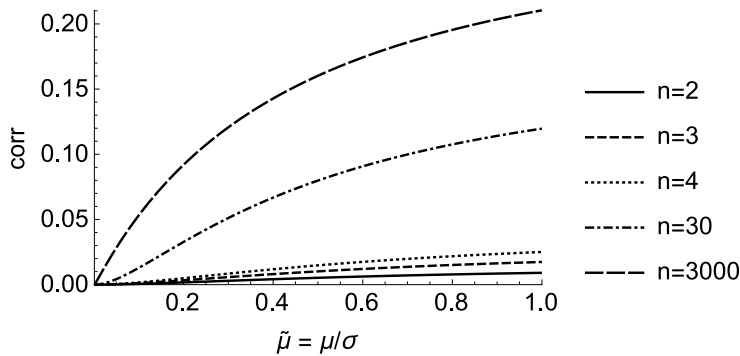}
\caption{Correlation coefficient as a function of the ratio $\tilde\mu$ for $n=2,3,4$, $n=30$, and $n=3000$ and $\tilde\rho=10$. The corresponding fraction \added{$F$} of free binding sites is monotonically decreasing from $\sim 0.41$ for \deleted{$\mu/\sigma$} \added{$\tilde \mu$}=1 to $\sim 0.27$ for $\tilde\mu=0$.}
\label{fig:correlation}
\end{figure}

\deleted{So while we find that the correlation coefficient is always strictly positive, in our numerical investigations, we find that, in general, the correlation coefficient is close to zero.}

\section{The mean field at distance one assumption and the spread of an \added{SIR} infectious disease on the network}\label{sec:meanfield}
In this text we have so far considered the static configuration network and two dynamic variants, one without and one with demographic turnover (Sections~\ref{sec:nodemography} and~\ref{sec:demography}, respectively). If we describe the network by only labelling individuals by their degree, i.e.\ by their numbers of partners, then no information about the partners of partners is included in the bookkeeping. However, in both the static network and the dynamic network without demography, we know that there is independence in the degrees of partners. Therefore, although not explicitly included in our bookkeeping, statistical information about the number of partners of partners of an individual with $k$ partners is readily available in the form of the size-biased degree distribution. 

This is not the case for the dynamic network model with demographic turnover. We have seen in Section~\ref{sec:demography} that in this dynamic network model there is dependence in degrees of partners. Therefore, the degree of the individual under consideration cannot be ignored when considering the degree of the partner. But\deleted{,} in previous work~\cite{Leung2015a, Leung2016}\deleted{,} we did exactly that. We pretended that the degree of the partner was independent of the degree of the focus individual, and we termed the approximation the \added{`}mean field at distance one assumption\added{'} (more appropriately we should have called it the mean field at distance one \emph{approximation}). \added{This mean field at distance one assumption allowed us to formulate a model for the spread of infection on the dynamic network with demography. It enabled us to write down a closed system of equations that is analytically tractable. Therefore, while we were well aware that an approximation was made by this assumption, it was very convenient to do so.} \deleted{While we were well aware that an approximation was made by this assumption, it was very convenient to do so. It allowed us to close the system of equations describing the disease dynamics on the network.} 

Information about partners of partners of an individual is crucial in describing the disease dynamics on the network. In the remainder of this section we elaborate on this point. Consider the spread of an SIR (Susceptible$\to$Infectious$\to$Recovered) infection superimposed on the network. We label each individual by (i) its disease status, (ii) the number of partners it has, and (iii) the disease status of each of these partners. Then, in order to describe the disease dynamics in the population, we need to make statements about partners of partners of an individual. Indeed, suppose we have an individual with a susceptible partner, then the rate at which this susceptible partner becomes infected depends on the total number of infectious partners it has. However, this kind of information is exactly what we do not have in our description. In fact, such information would only be available if we have a complete description of the entire network. Indeed, suppose we would incorporate partners of partners in our description. Then we would also need to know about the number of infectious partners of susceptible partners of partners, etcetera. This is where the mean field at distance one assumption comes into play. This is the assumption that we may average over the population in a certain way and consider the \emph{expected} number of infectious partners of a susceptible partner instead. 

In the static network and the dynamic network without demography, with independence in degrees, averaging is done as follows. Consider an individual $u$ with disease status $d$ (either susceptible, infectious, or recovered), and consider a susceptible partner $v$ of $u$. Then we take into account that $u$ has disease status $d$ but not the degree of $u$ (here we use independence of degrees): the expected number of infectious partners of $v$ is the expected number of infectious partners of a susceptible partner of an individual with disease status $d$. The latter is the expected number at the p-level. In fact, we can apply the mean field at distance one assumption purely at the binding site level. The probability that $v$ has $k$ partners is given by the size-biased degree distribution. Individual $v$ has one special binding site for which the transmission rate along this binding site is determined by the disease status of $u$. The $k-1$ other binding sites of $v$ are indistinguishable. The transmission rate along each of these other $k-1$ binding sites is determined by the probability that a binding site of $v$ is occupied by an infectious partner.

In the static network case one can prove that the mean field at distance one description is exact: the deterministic description can be obtained as the large population limit of a stochastic model (under suitable technical conditions), see~\cite{Decreusefond2012,Barbour2013,Janson2014}. For the dynamic network without demography we conjecture that this is also true and we pose this as an open problem in~\cite{Leung2016}. In Section~\ref{sec:nodemography} we provided evidence in support of this conjecture by showing that, as in the static setting, there is independence in the degrees of partners. In the dynamic network with demography we \emph{know} we need to take into account the dependence in degrees. 

In Section~\ref{sec:degreedegree} we quantified the dependence between degrees through the correlation coefficient that we subsequently studied numerically. While the degree correlation is always larger than zero in case of demographic turnover, in general, we found the correlation coefficient to be quite small. So, even though an approximation is made by ignoring knowledge about the degree of an individual $u$ when considering the degree of a partner $v$, this approximation may not be that bad. Yet, we wonder whether it is possible to give an exact statistical description of the disease dynamics on the network. We conjecture that this is in fact possible and that the key to this is age. 

As we have seen in Section~\ref{sec:demography}, degree-degree correlation can be deduced from age-age correlation. If we incorporate the age of partners in the bookkeeping of individuals, we can use that age to predict the number of partners of those partners. More concretely, consider a binding site belonging to an individual with age $a$. This binding site is either free or occupied by a partner \emph{with age $\alpha$}. If the binding site is, at age $a$, occupied by a partner with age $\alpha$, then this partner has $k$ partners with probability $q_k(\alpha)$, where $q_k(\alpha)$ is given by \added{eq.}~\eqref{eq:qk_age}. In particular, no approximation needs to be made. We conjecture that this carries over to the setting with an infectious disease superimposed on the network: by including in our bookkeeping not only the disease status and age \deleted{(and time of birth)} of the binding site under consideration and the disease status of \deleted{the partner (if it has one) but also the age of this partner} \added{any partner, but also the ages of these partners}, one can again employ the mean field at distance one assumption without making an approximation, i.e.\ average over the population in the correct way: we may consider the \emph{expected number} of infectious partners of a susceptible partner \emph{of age $\alpha$} of an individual with disease status $d$ and age $a$ (which is part of the description of the model if the bookkeeping includes age of partners).

Proving this claim about the bookkeeping with the age of partners included is \added{both} outside the scope of this text \added{and outside our area of expertise}. Rather we conclude this section by highlighting some aspects of the mean field at distance one assumption by considering the basic reproduction number $R_0$. The traditional perspective that one takes for $R_0$ is that of an infectious case: $R_0$ can be interpreted as the expected number of secondary cases generated by one typical newly infected case at the beginning of an epidemic. As we explained in~\cite{Leung2016}\added{,} it can be advantageous to take \deleted{a}\added{the} different perspective of `reproduction opportunities' (where `reproduction' corresponds to transmission of the infectious agent to another host). In this context reproduction opportunities consist of $-+$ links, i.e.\ partnerships between susceptible ($-$) and infectious ($+$) individuals. This different perspective does not change the expression that one obtains for $R_0$. So we can interpret $R_0$ as the expected number of $-+$ links generated by one typical newly formed $-+$ link at the beginning of an epidemic. 

The reasoning in~\cite[Section 4.3]{Leung2016} was as follows. \added{At the beginning of an epidemic, for an SIR infection,} there are two birth-types of $-+$ links:
\begin{description}
\item[Type 0] the $-+$ link was formed when a $-$ binding site and a $+$ binding site got connected
\item[Type 1] the $-+$ link is a transformed $--$ link (one of the two owners got infected by one of its other partners)
\end{description}
Note that the density of the age distribution of the owner of a binding site is $\pi_0$ upon partner formation (see \added{eq.}~\eqref{eq:agedistribution}). Therefore, the density of the age distribution of the $-$ binding site of Type 0 $-+$ links is $\pi_0$. However, the age of the $-$ binding site of Type 1 $-+$ links is correlated to the age of its $+$ partner. In~\cite{Leung2016}, we approximated the age of the $-$ binding site in the Type 1 link by ignoring the correlation with the age of its $+$ partner. We approximated the density of the age distribution of the $-$ binding site in the Type 1 link by $\pi_1$ where $\pi_1$ is given by \added{eq.}~\eqref{eq:agedistribution2} (see~\cite[Section 4.3]{Leung2016} for details). These densities $\pi_0$ and $\pi_1$ for the ages of binding sites are key in characterizing $R_0$. One ends up with a characterization of $R_0$ as the dominant eigenvalue of a $2\times2$ next-generation matrix $K$ \added{where entry $K_{ij}$ of $K$ can be interpreted as the expected number of secondary cases with state-at-infection $i$ caused by one newly infected individual with state-at-infection $j$ at the beginning of an epidemic}~\added{\cite[Chapter 7]{Diekmann2013}}. 

However, if we include the age of partners in our bookkeeping, then this also needs to be included in our characterization of $R_0$. While nothing changes for the $-$ binding sites in the Type-0 links in terms of the density of the age distribution, we can no longer simply consider Type-1 links. Rather, one needs to keep track of age at the moment that the Type 1 $-+$ link is born. This leads to an infinite-dimensional problem rather than \added{the simple setting of two types} \deleted{the very nice two types} that arise\deleted{s} from the approximation. Clearly from the point of view of the \deleted{$R_0$ characterization} \added{characterization of $R_0$} it is attractive to make an approximation by ignoring age correlation \deleted{and assuming independence instead} \added{i.e.\ assuming independence}. One only deals with two types (and the dominant eigenvalue of a $2\times2$ next-generation matrix) rather than infinitely many types \added{(and a corresponding next-generation operator and its spectral radius)}. 

Is \deleted{all} \added{the mathematical tractability} then lost by including ages of partners in the bookkeeping? No, not necessarily. But the \mbox{$R_0$-characterization} does illustrate that \added{including the age of partners in the bookkeeping will make the model formulation and analysis far} \deleted{things will be} less straightforward than in the static network or the dynamic network without demography. 

\section{Conclusion}\label{sec:conclusion}
In this text we discussed the mean field at distance one assumption for two dynamic network models \added{of~\cite{Leung2016}} that are inspired by the (static) configuration network. The first dynamic network model includes partnership formation and separation, while the second dynamic network model additionally includes demographic turnover. We concerned ourselves with a description that only includes individuals and their partners, without any information about partners of partners in the bookkeeping. The mean field at distance one assumption concerns itself with these partners of partners. It states that one can average over the population in a well-defined way to obtain the relevant information. In case of a static configuration network the mean field at distance one assumption holds as there is independence in the degrees of partners. This independence in degrees of partners is shared by the dynamic network model without demography\added{; see Section~\ref{sec:nodemography}}. \added{This independence result suggests that we can describe the spread of infection on the dynamic network without demography using the mean field at distance one approach (as indeed conjectured in~\cite{Leung2016}).} \deleted{The results of Section~\ref{sec:nodemography} are in favour of our previous conjecture that our model for the spread of infection on the dynamic network without demography is exact; see~\cite{Leung2016}.}

However, degree dependence between partners arises in the dynamic network model with demography. We showed this via the existing age dependence between partners \added{in Section~\ref{sec:ageage}}. \added{As discussed in Section~\ref{sec:meanfield},} \added{i}\deleted{I}n previous work we ignored these dependencies between the partners~\cite{Leung2016,Leung2015a}. In \added{the current} \deleted{this} text we investigated the dependency between partners by means of the degree correlation coefficient. In general this degree correlation coefficient is positive but reasonably small. This is especially the case if demographic and partnership changes are on somewhat different time scales, and \deleted{partner acquisition} \added{partnership formation} and separation are on comparable time scales, and partnership capacity $n$ is not too large, which are quite reasonable assumptions to make.

Clearly there are advantages to approximating the true process by ignoring these dependencies between partners. Especially if the degree correlation is rather small, then it is attractive to do so. The goal of this text is not to advocate that one should never concern oneself with approximations (clearly not as this is exactly what we have done in previous work). Rather, our point is that it is important to be aware of the assumptions that one makes when formulating models and the limitations and consequences of the assumptions. 

Ideally, one can provide a statistical description for transmission dynamics on a network without making approximations (whether it is desirable to still make approximations, e.g.\ for computational convenience, is a different issue). In Section~\ref{sec:meanfield}, we speculated that by incorporating age of partners in the dynamic network with demography, one can avoid making approximations. But, as we also outlined in the same section, this probably comes at a price. It may be that the analysis of the model becomes much harder. How to formulate and analyse the model that includes ages of partners is outside of the scope of this text and is left for future work. Here we end with the conjecture that bookkeeping that takes the age of partners into account allows for an exact description of the spread of infectious diseases on the dynamic network with demography. \added{We hope that this text motivates some probabilists to take up the challenge of proving (or, unexpectedly, disproving) the conjecture.}

\appendix

\setcounter{equation}{0}
\renewcommand{\theequation}{\Alph{section}.\arabic{equation}}
\numberwithin{equation}{section}

\section{Relationship between $\boldsymbol{\pi_1(a)}$ and $\boldsymbol H$}\label{app:probint}
The probability density function $\pi_1(a)$ for the age of an occupied binding site (see \added{eq.}~\eqref{eq:agedistribution2}) is related to the probability density function $H(a,\alpha)$ for the ages of two partners $u$ and $v$ in a randomly chosen partnership. This relation is formulated in \added{eq.}~\eqref{eq:probint}.

\begin{lemma}\label{lem:probint}
\begin{equation}\label{eq:probint}
\int_{\alpha=0}^\infty H(a,\alpha)d\alpha=\pi_1(a).
\end{equation}
\end{lemma}

\begin{proof}
First note that we can rewrite $H(a,\alpha)$ as
\begin{equation*}
H(a,\alpha)=\frac{\mu e^{-\mu a}}{1-F}\int_{\xi=0}^{\min{(a,\alpha)}}\rho F\varphi(a-\xi)\pi_0(\alpha-\xi)e^{-(\sigma+\mu)\xi}d\xi\added{.}
\end{equation*}
Next, one finds that 
\begin{equation}\label{lemint}
\int_{\alpha=0}^\infty\int_{\xi=0}^{\min{(a,\alpha)}}\rho F\varphi(a-\xi)\pi_0(\alpha-\xi)e^{-(\sigma+\mu)\xi}d\xi d\alpha=1-\varphi(a),
\end{equation}
\deleted{either} by direct calculations using \added{eqs.}~\eqref{varphi} and~\eqref{eq:agedistribution2} for $\varphi$ and $\pi_0$, and we conclude that \added{eq.}~\eqref{eq:probint} holds. One can also reason as follows for \added{eq.}~\eqref{lemint}: $\rho F \varphi(a-\xi)\pi_0(\alpha-\xi)e^{-(\sigma+\mu)\xi}$ is the probability that a binding site with age $a$ has a partner with age $\alpha$ and partnership duration $\xi$ given that the owner of the binding site under consideration does not die. By integrating over all possible partnership durations $0\leq\xi\leq\min(a,\alpha)$, we obtain the probability that a binding site with age $a$ has a partner with age $\alpha$ (given that the owner does not die): $\int_0^{\min(a,\alpha)} \rho F \varphi(a-\xi)\pi_0(\alpha-\xi)e^{-(\sigma+\mu)\xi}d\xi$. Then finally, by integrating over all possible $\alpha\geq0$ we obtain the probability $1-\varphi(a)$ that a binding site with age $a$ is occupied.
\end{proof}

\section{The marginal degree distribution $\B{Q_k=\sum_lP(k,l)}$}\label{app:marginal}
We obtain the probability $Q_k$ that an individual involved in a randomly chosen partnership has degree $k$ from the joint probability distribution $P(k,l)$\added{, cf.~\eqref{exact},} by summing over all $l=1,\ldots, n$, i.e.\ $Q_k=\sum_{l=1}^n P(k,l)$. On the other hand, in previous work~\cite{Leung2012} we have derived an expression for $Q_k$ from the stable degree distribution $(P_k)_k$ in the population: $Q_k=k P_k/ n(1-F)$. We show that both ways of arriving at $Q_k$ yield the same expression, i.e.\
\begin{equation}\label{eq:Qk}
Q_k=\sum_{l=1}^n P(k,l)=k P_k/ n(1-F)\added{.}
\end{equation}
First, we work out the right-hand side. $P_k$ is expressed in terms of the probability $\varphi(a)$ as follows:
\begin{equation*}
P_k=\binom nk \int_0^\infty\mu e^{-\mu a}\varphi(a)^{n-k}(1-\varphi(a))^{k}da.
\end{equation*}
On the other hand, we can simplify $\sum_{l=1}^n P(k,l)$. First of all, note that since $q_l(a)$ is a probability distribution, $\sum_{l=1}^n q_l(\alpha)=1$. Therefore
\begin{align}
\sum_{l=1}^n P(k,l)&=\sum_{l=1}^n\int_{a=0}^\infty\int_{\alpha=0}^\infty\int_0^{\min(a,\alpha)} q_k(a)q_l(\alpha)\frac{\rho F^2\pi_0(a-\xi)\pi_0(\alpha-\xi)e^{-(\sigma+2\mu)\xi}}{1-F}d\xi d\alpha da\nonumber\\
&=\int_{a=0}^\infty\int_{\alpha=0}^\infty\int_0^{\min(a,\alpha)} q_k(a)\frac{\rho F^2\pi_0(a-\xi)\pi_0(\alpha-\xi)e^{-(\sigma+2\mu)\xi}}{1-F}d\xi d\alpha da.\label{eq:simplify}
\end{align}
Next, note that we can simplify \added{eq.}~\eqref{eq:simplify} as follows:
\begin{align*}
\sum_{l=1}^n P(k,l)&=\frac{k\binom nk}{n(1-F)}\int_{a=0}^\infty\mu e^{-\mu a}\varphi(a)^{n-k}(1-\varphi(a))^k(1-\varphi(a))^{-1}\\
&\phantom{=\quad}\int_{\alpha=0}^\infty\int_{\xi=0}^{\min(a,\alpha)}\rho F^2 e^{\mu\xi}\frac{\varphi(a-\xi)}{F}\frac{\mu e^{-\mu(\alpha-\xi)}\varphi(\alpha-\xi)}{F}e^{-(\sigma+2\mu)\xi}d\xi d\alpha da\\
&=\frac{k\binom nk}{n(1-F)}\int_{a=0}^\infty\mu e^{-\mu a}\varphi(a)^{n-k}(1-\varphi(a))^k(1-\varphi(a))^{-1}\\
&\phantom{=\quad}\int_{\alpha=0}^\infty\int_0^{\min(a,\alpha)} \rho F \varphi(a-\xi)\pi_0(\alpha-\xi)e^{-(\sigma+\mu)\xi}d\xi d\alpha da\\
&=\frac{k\binom nk}{n(1-F)}\int_{a=0}^\infty\mu e^{-\mu a}\varphi(a)^{n-k}(1-\varphi(a))^k(1-\varphi(a))^{-1}(1-\varphi(a))da.
\end{align*}
Here we used \added{eq.}~\eqref{eq:probint} in the third equality. So we find that \added{eq.}~\eqref{eq:Qk} indeed holds.

\section{The correlation coefficient as a function of model parameters}\label{app:corr}
We work out \added{eq.}~\eqref{eq:corr} by \deleted{working out}\added{computing} $A$, $B$, and $C$, \deleted{where these are} defined by \added{eq.}~\eqref{eq:sums}.
\begin{align}
A&=\sum_{k=1}^n\sum_{l=1}^nk lP(k,l)=\int_{a=0}^\infty\int_{\alpha=0}^\infty\sum_{k=1}^nk q_k(a)\sum_{l=1}^n l q_l(\alpha)H(a,\alpha) dad\alpha\nonumber\\
&=\int_{a=0}^\infty\int_{\alpha=0}^\infty\big(n(1-\varphi(a))+\varphi(a)\big)\big(n(1-\varphi(\alpha))+\varphi(\alpha)\big)H(a,\alpha)dad\alpha\nonumber\\
&=\frac{1}{(1-F) (\rho F+\sigma+2 \mu)^2 (2 \rho F +3 \sigma+4 \mu ) (2 (\rho F +\sigma +\mu)+ \mu)^2}\nonumber\\
&\phantom{=\ }\Big\{\rho(\mu ^2 \left(\rho^2F^2 (n (33 n+46)+1)+2 \rho F \sigma(92 n+33) +179 \sigma ^2)\right.\nonumber\\
&\phantom{=\ }\quad\left.+4 \mu  (\rho F n +\sigma ) (4\rho ^2 F^2 n +10\rho F \sigma(n+1) +19 \sigma ^2\right)\nonumber\\
&\phantom{=\ }\quad+6 \mu ^3 (2 \rho F (8 n+3)+31 \sigma )+4 \sigma  (2 \rho F+3 \sigma ) (\rho F n +\sigma )^2+72 \mu ^4)\Big\}\label{corr1},
\end{align}
where the last equality is calculated using Mathematica.

We already calculated the mean $B=\sum_{k=1}^nkQ_k$ in~\cite[eq.~(23)]{Leung2012}. For completeness, we work it out using the probability distribution $(P(k,l))$.
\begin{align}
B&=\sum_{k=1}^nkQ_k=\sum_{k=1}^nk\sum_{l=1}^nP(k,l)\nonumber\\
&=\int_{a=0}^\infty\sum_{k=1}^nk q_k(a)\int_{\alpha=0}^\infty\sum_{l=1}^nq_l(\alpha)H(a,\alpha)d\alpha da\nonumber\\
&=\int_{a=0}^\infty\big(n(1-\varphi(a))+\varphi(a)\big)\int_{\alpha=0}^\infty H(a,\alpha)d\alpha da\nonumber\\
&=\int_{a=0}^\infty\big(n(1-\varphi(a))+\varphi(a)\big)\pi_1(a)da\nonumber\\
&=1+\frac{2\rho F(n-1)}{2(\rho F+\sigma+\mu)+\mu}.\label{corr2}
\end{align}
In the first equality we used that $(Q_k)$ is the marginal distribution of $(P(k,l))$, and in the fifth equality we used identity~\eqref{eq:probint}.

Finally, we consider the second moment $C=\sum_{k=1}^nk^2Q_k$:
\begin{align}
C&=\sum_{k=1}^nk^2Q_k=\sum_{k=1}^nk^2\sum_{l=1}^nP(k,l)\nonumber\\
&=\int_{a=0}^\infty\sum_{k=1}^nk^2 q_k(a)\int_{\alpha=0}^\infty\sum_{l=1}^nq_l(\alpha)H(a,\alpha)d\alpha da\nonumber\\
&=\int_{a=0}^\infty\big(n^2(1-\varphi(a))^2+\varphi(a)((3n-1)(1-\varphi(a))+\varphi(a))\big)\int_{\alpha=0}^\infty H(a,\alpha)d\alpha da\nonumber\\
&=\int_{a=0}^\infty\big(n^2(1-\varphi(a))^2+\varphi(a)((3n-1)(1-\varphi(a))+\varphi(a))\pi_1(a)da\nonumber\\
&=\frac{\rho F (\mu(12\mu+\rho F(24 n-7)+17\sigma)+6(\rho^2 F^2n^2+\rho F(3n-1)\sigma+\sigma^2))}{(1-F)(\rho F+\sigma+2\mu)(2(\rho F+\sigma+\mu)+\mu)(3(\rho F+\sigma+\mu)+\mu)}\added{.}\label{corr3}
\end{align}

\deleted{Putting}\added{Inserting} \added{eqs.}~\eqref{corr1},~\eqref{corr2}, \added{and}~\eqref{corr3} together in \added{eq.}~\eqref{eq:corr}, we find an explicit expression for the correlation coefficient $corr$. Note that the variance of a random variable is always nonnegative (and nonzero if \added{the random variable is} not equal to a constant). Therefore, we find that the sign of $Cov(D_u, D_v)=A-B^2$ determines the sign of the correlation coefficient $corr$ in \added{eq.}~\eqref{eq:corr}. Note that identity~\eqref{eq:Fconstant} for $F$ allows us to express $\sigma$ in terms of the other parameters: $\sigma=\rho F^2/(1-F)-2\mu$. For clarity, we use this identity for $\sigma$ in the numerator (but not in the denominator) in the simplification of $Cov(D_u, D_v)$. We find that
\begin{equation}\label{eq:covariance}
A-B^2=\frac{\mu^2\rho^3F^2(1-F)(n-1)^2}{(1-F)^2(\rho F+\sigma+2\mu)^2(2\rho F+3\sigma+4\mu)(2\rho F+2\sigma+3\mu)^2}.
\end{equation}
In particular, the covariance (and therefore the correlation coefficient $corr$) is strictly larger than zero if $\rho>0$, $\sigma>0$, $\mu>0$, and $n>1$. 

\subsubsection*{Acknowledgements}
\added{We would like to thank Pieter Trapman for opening our eyes during the Infectious Disease Dynamics meeting at the Isaac Newton Institute in Cambridge in 2013 as well as the members of the infectious disease dynamics journal clubs in Utrecht and Stockholm, and two anonymous reviewers for helpful comments.}

K.Y. Leung is supported by the Netherlands Organisation for Scientific Research (NWO) [grant Moza\"iek 017.009.082] and the Swedish Research Council [grant number 2015-05015\_3].

\setlength{\bibsep}{2pt}
\bibliographystyle{unsrt}
\bibliography{refs_network}

\begin{thebibliography}{10}

\bibitem{Leung2016}
K.~Y. Leung and O.~Diekmann.
\newblock Dangerous connections: on binding site models of infectious disease
  dynamics.
\newblock {\em J. Math. Biol.}, 74:619--671, 2017.

\bibitem{Barbour2013}
A.~D. Barbour and G.~Reinert.
\newblock {Approximating the epidemic curve}.
\newblock {\em Electron. J. Probab.}, 18(54):1--30, 2013.

\bibitem{Volz2008}
E.~M. Volz.
\newblock {SIR dynamics in random networks with heterogeneous connectivity}.
\newblock {\em J. Math. Biol.}, 56:293--310, 2008.

\bibitem{Kiss2017}
I.~Z. Kiss, J.~C. Miller, and P.~Simon.
\newblock {\em Mathematics of epidemics on networks: from exact to approximate
  models}.
\newblock Springer, 2017.

\bibitem{Durrett2006}
R.~Durrett.
\newblock {\em {Random graph dynamics}}.
\newblock Cambridge University Press, 2006.

\bibitem{vdHofstad2015}
R.~{van der Hofstad}.
\newblock {\em {Random graphs and complex networks Vol.\ I}}.
\newblock Lecture notes, Available: \url{http://www.win.tue.nl/~rhofstad/},
  September 2015.

\bibitem{Leung2012}
K.~Y. Leung, M.~E.~E. Kretzschmar, and O.~Diekmann.
\newblock {Dynamic concurrent partnership networks incorporating demography.}
\newblock {\em Theor. Popul. Biol.}, 82:229--239, 2012.

\bibitem{Britton2010}
T.~Britton and M.~Lindholm.
\newblock {Dynamic random networks in dynamic populations}.
\newblock {\em J. Stat. Phys.}, 139:518--535, 2010.

\bibitem{Britton2011}
T.~Britton, M.~Lindholm, and T.~Turova.
\newblock {A dynamic network in a dynamic population: asymptotic properties}.
\newblock {\em J. Appl. Prob.}, 48:1163--1178, 2011.

\bibitem{Lashari2016}
A.~A. Lashari and P.~Trapman.
\newblock Branching process approach for epidemics in dynamic partnership
  network.
\newblock {\em Submitted}.

\bibitem{Leung2015a}
K.~Y. Leung, M.~E.~E. Kretzschmar, and O.~Diekmann.
\newblock {$SI$ infection of a dynamic partnership network: characterization of
  $R_0$}.
\newblock {\em J. Math. Biol.}, 71:1--56, 2015.

\bibitem{Newman2002}
M.~E.~J. Newman.
\newblock {Assortative Mixing in Networks}.
\newblock {\em Phys. Rev. Lett.}, 89(20):208701, 2002.

\bibitem{Newman2003}
M.~E.~J. Newman.
\newblock {The structure and function of complex networks}.
\newblock {\em SIAM Review}, 45(2):167--256, 2003.

\bibitem{Ball2012}
F.~Ball, T.~Britton, and D.~Sirl.
\newblock {A network with tunable clustering, degree correlation and degree
  distribution, and an epidemic thereon}.
\newblock {\em J. Math. Biol.}, 66:979--1019, 2013.

\bibitem{Decreusefond2012}
L.~Decreusefond, J.-S. Dhersin, P.~Moyal, and V.~C. Tran.
\newblock {Large graph limit for an SIR process in random network with
  heterogeneous connectivity}.
\newblock {\em Ann. Appl. Probab.}, 22:541--575, 2012.

\bibitem{Janson2014}
S.~Janson, M.~Luczak, and P.~Windridge.
\newblock {Law of large numbers for the SIR epidemic on a random graph with
  given degrees}.
\newblock {\em Random Struct. Algor.}, 45(4):724--761, 2014.

\bibitem{Diekmann2013}
O.~Diekmann, J.~A.~P. Heesterbeek, and T.~Britton.
\newblock {\em {Mathematical tools for understanding infectious disease
  dynamics}}.
\newblock Princeton University Press, 2013.

\end{thebibliography}

\end{document}